\documentclass[10pt, aps, floatfix, pra,twocolumn, superscriptaddress, showpacs, noshowkeys, english]{revtex4-1}
\usepackage{amsthm, amssymb, amsmath, amsfonts}
\usepackage{graphicx}
\usepackage{epsfig}
\usepackage{footnote}
\usepackage[english]{babel}
\newtheorem{thm}{Theorem}
\newtheorem{lem}[thm]{Lemma}
\newtheorem*{thm*}{Theorem}
\DeclareMathAlphabet{\mathbbm}{U}{bbm}{m}{n}
\SetMathAlphabet\mathbbm{bold}{U}{bbm}{bx}{n}
\newcommand{\1}{\mathbbm{1}}   
\newcommand{\bea}{\begin{eqnarray}}
\newcommand{\eea}{\end{eqnarray}}
\newcommand{\tr}[1]{{\rm tr}\left[#1\right]}
\newcommand{\C}{\mathbb{C}}
\newcommand{\E}{\mathbb{E}}
\newcommand{\EE}{\mathcal{E}}
\newcommand{\FF}{\mathcal{F}}

\renewcommand{\>}{\rangle}
\newcommand{\<}{\langle}

\begin{document}

\title{Entanglement, fractional magnetization and long-range interactions}
\author{Andrea Cadarso}
\affiliation{Instituto de F\'isica Fundamental, IFF-CSIC, Serrano 113-bis, 28006 Madrid, Spain}
\affiliation{Departamento de An\'{a}lisis Matem\'{a}tico, Universidad Complutense de Madrid, 24040 Madrid, Spain}
\author{Mikel Sanz}
\affiliation{Max Planck Institut f\"{u}r Quantenoptik, Hans-Kopfermann-Str. 1, 85748 Garching, Germany}
\author{Michael M. Wolf}
\affiliation{Department of Mathematics, Technische Universit\"{a}t M\"{u}nchen, 85748 Garching, Germany}
\author{J. Ignacio Cirac}
\affiliation{Max Planck Institut f\"{u}r Quantenoptik, Hans-Kopfermann-Str. 1, 85748 Garching, Germany}
\author{David P\'{e}rez-Garc\'{\i}a}
\affiliation{Departamento de An\'{a}lisis Matem\'{a}tico, Universidad Complutense de Madrid, 24040 Madrid, Spain}



\begin{abstract}
Based on the theory of Matrix Product States, we give precise statements and complete analytical proofs of the following claim: a large fractionalization in the magnetization or the need of long-range interactions imply large entanglement in the state of a quantum spin chain. 
\end{abstract}

\maketitle

\section*{Introduction} Entanglement plays a central role in many-body quantum systems as it can be used to understand the structure of the quantum states that appear in nature. In systems governed by short-range interactions, low energy states possess very little entanglement. In contrast, states evolved after quenches display large amounts of entanglement. These different behaviors, which are supported by abundant numerical evidence, have been recently established on solid grounds in one spatial dimension \cite{quenches}. In particular, ground states of gapped (critical) Hamiltonians fulfill an area law, in which the entanglement entropy of any connected region is bounded by a constant (diverges at most like the logarithm of the number of spins in that region) \cite{entropyandarea, area-law}. These results immediately imply that the ground state of a spin chain can be well approximated by Matrix Product States \cite{PVWC07, FNW92}, and, thus, such family of states captures the physics in one dimension \cite{hastings-mps, hastings-mps-2, mps-vidal}.

Apart from the cases mentioned above, there exists practically no other physical situation where the existence of large or small amounts of entanglement can be rigorously established. In this paper, we identify two scenarios in one spatial dimension that can be connected to the presence of entanglement. We will restrict ourselves to systems described by MPS, and thus, our results do not apply to general situations. Nevertheless, given the fact that such family of states approximates well one dimensional systems, we conjecture that our results are true in more general settings. 

The first scenario corresponds to the presence of fractionalization, a striking phenomenon that arises whenever certain observables, which are expected to take integer expectation values, appear to be fractional--valued instead.
The most prominent example of such behavior is the celebrated fractional quantum Hall effect \cite{stormer, tsui, laughlin}. In recent years, this phenomenon has been extensively studied in many systems, including spin chains \cite{Oshikawa,hida}, where the magnetization per particle is fractionalized as a function of the external magnetic field. In the first part of this manuscript, we establish a lower bound for the entanglement entropy of any (connected and sufficiently large) region of a quantum spin chain in terms of the fractionalized magnetization.

The second scenario corresponds to a situation where the area law does not apply, namely when studying a spin chain with long-range interactions. Intuitively, one can expect that such interactions give rise to large amounts of entanglement since any specific region will be correlated to any other part of the chain. However, it is very subtle to transform this intuitive idea into a rigorous result. The main reason is that the ground state of Hamiltonian containing long-range interactions may coincide with (or be very similar to) the ground state of another Hamiltonian containing short-range ones and, therefore, fulfilling the area law. For instance, if we have an Ising model with decaying interactions (in the absence of a transverse magnetic field), the ground state will be still a product state, which, in turn, is also the ground state of the Ising model with nearest-neighbor couplings. Such state does not display any entanglement at all. Hence, we can only expect to have large amounts of entanglement whenever such examples do not exist; that is, whenever our state is (in some sense) not close to any other state corresponding to the ground state of a Hamiltonian with short-range interactions. In fact, we will prove a Theorem that formalizes this statement in the second part of this paper. 

In order to rigorously prove our statements, we will have to further develop the theory of MPS, extending previous results presented in Refs.\ \cite{PVWC07, SWPC09,SWPC09-2}, and deriving new ones. Some of them are very intuitive, although the rigorous proofs are somewhat complicated. We will present in the main text of this paper the main steps and their intuitive interpretation, and leave for the appendices the technical details.

\section*{Matrix Product States} This family of states describes a chain of $N$ spins $J$, with $d=2J+1$, and can be written as
\begin{equation}\label{eq:General-MPS}
|\psi\>=\sum_{i_1,\ldots, i_N=1}^d\tr{A_{i_1}[1]\cdots A_{i_N}[N]}|i_1\cdots i_N\>
\end{equation}
Here, $A_i [n]$ are $D\times D$ matrices associated to the spin in the $n$-th site of the chain. Our results, unless specifically mentioned, concern translationally invariant states, where $A_i [n]=A_i$ independently of the site $n$. We will call the corresponding state $|\psi_A\>$.

Let us recall some known properties of such states (see, for instance, Refs. \cite{FNW92,PVWC07}). MPS can be classified into injective and non-injective. An MPS is called {\it injective} if there exists an $L$ such that for regions of size $L$ or larger, different boundary conditions give rise to different states; that is, the map $\Gamma(X)=\sum_{i_1,\ldots, i_L}\tr{XA_{i_1}\cdots A_{i_L}}|i_1\cdots i_L\>$ is injective. This is known to be equivalent to the fact that, after a suitable transformation of the form $A_i\mapsto XA_iX^{-1}$, one obtains a {\it canonical form} fulfilling the following conditions (that we will always assume for injective MPS): (i) $\sum_{i}A_i A_i^\dagger=\1$, (ii) $\sum_i A_i^\dagger \Lambda_A A_i=\Lambda_A$ for a diagonal positive full rank matrix $\Lambda_A$, and (iii) the  cp map $\E_A$ defined as 
 \begin{equation}
 \E_A(X)=\sum_i A_i XA_i^\dagger 
 \end{equation}
has 1 as its unique non-degenerate eigenvalue of maximal modulus. This canonical form is unique in the following sense: if $A$ and $B$ are matrices giving rise to {\it different} canonical representations of the {\it same} MPS, then they must be related by a unitary $U$ according to $e^{i\theta}A_i=UB_iU^\dagger$. {\it Non-injective} MPS also possess a canonical form where the matrices are block-diagonal and the cp map associated to each block verifies conditions (i), (ii) and (iii) above, except for the existence of other eigenvalues of modulus 1.

\section*{Large fractional magnetization implies large entanglement} 
Fractional magnetization in a spin chain occurs whenever we have $U(1)$ symmetry (generated, in this case, by an operator $J_z$), and the expectation value of the generator $m=\langle J_z\rangle/N,$ the magnetization per particle, fulfills that $J-m=q/p$, where $p$ and $q$ are coprime. When we change some external parameter, such as a magnetic field, the value of $m$ generally changes in discrete steps, giving rise to typical plateaus in the magnetization. Our aim is to show that whenever a translationally and $U(1)$ invariant MPS displays this phenomenon, the entanglement entropy of any sufficiently large region is greater than $\log(p)$. That is, the value of $p$ imposes some lower bound on the entanglement present in the system. We will start out with a trivial example that will help us build an intuition about this statement, and then we will generalize this claim to arbitrary states.

Let us consider $J=1/2$, any two numbers $q,$ $p$ coprime, and construct a state of $N=np$ ($n$ integer) spins as follows. We consider first a $p$-particle state of the form $|a\>=|\uparrow\uparrow\cdots \uparrow\downarrow\downarrow\cdots \downarrow\>$, where $q$ is the number of spins down. Then, we take $n=N/p$ copies of such state, and build an equal superposition of the $p$ possible different translations of $|a\>^{\otimes n},$
\begin{equation}\label{eq:example}
|\psi\>=\frac{1}{\sqrt{p}}\sum_{m=0}^{p-1} {\tau}^m |a\>^{\otimes n}
\end{equation}
where $\tau$ is the translation operator. This state is translationally invariant, has $U(1)$ symmetry generated by $J_z=\sum s_n^z$, where $s_z$ is the single-spin operator $2s_z=\vert \uparrow\>\<\uparrow \vert - \vert \downarrow\>\<\downarrow \vert$, i.e.  $J_z |\psi\> = N (1/2-q/p)|\psi\>$, and, thus, exhibits fractional magnetization. Following the prescriptions of Oshikawa et al. \cite{Oshikawa}, this example contains ``periodic components'' in order to display such a phenomenon.
As one can see by simple inspection, if we take any region $A$ of size $L= kp,$ with $k \in \{1, \ldots, n \}$, the reduced density operator can be written as
\begin{equation}\label{eq:example_do}
\rho_L=\frac{1}{p}\sum_{m=0}^{p-1} |\varphi_m\>\<\varphi_m|
\end{equation}
where $\varphi_m$ are mutually orthogonal. Thus, the entropy of $\rho_L$ (and, consequently, the entanglement entropy) between the region $A$ and the rest is $\log(p).$ This toy model presenting such entropy is connected to the fact that, in this case, fractional magnetization arises because the ground state is a linear superposition of $p$-particle states which are both locally orthogonal (i.e. fully distinguishable) and related through a translation.

In what follows, we will consider the richer family of MPS in order to prove a related result. Note that the previous toy example is contained in the family of MPS just by considering the matrices
\begin{equation}
A_\downarrow=\sum_{i=1}^q |i\>\<i+1|\; ,\; A_\uparrow=\sum_{i=q+1}^{p}|i\>\<i+1|
\end{equation}
However, general cases of MPS possess several difficulties. In particular: (i) finding a characterization of all MPS displaying fractional magnetization; (ii) the fact that, even if an MPS is a superposition of states related by a translation, nothing ensures that the reduced states will be sums of few pure and mutually orthogonal states. In any case, we are able to prove the following Theorem:

\begin{thm}\label{thm:fractional}
Let $|\psi \rangle$ be a translational and $U(1)$ invariant MPS of spin $J$, with magnetization per particle $m$ and verifying $J-m = \frac{q}{p}$ ($p$ and $q$ coprime). Then there exists a constant $\gamma\in \mathbb{N}$ such that the entropy of the reduced density matrix of any region of size $L=k\gamma p\; (\forall k \in \mathbb{N})$ verifies $S(\rho_L) \geq \log (p),$ up to a exponentially small correction in $N-L$ and in $k.$
\end{thm}

In order to prove this, we proceed as follows: 

\begin{proof} Let $|\psi \rangle$ be an MPS, which is translational and $U(1)$ invariant. We also impose that this MPS has spin $J$ and magnetization per particle $m,$ verifying $J-m = \frac{q}{p}$ ($p$ and $q$ coprime) and consider its canonical form. If it has a single block, due to Lemma 8 in Appendix A, it must be $\gamma p$-periodic, where $\gamma \in \mathbb{N}$. This means that all the eigenvalues of magnitude one corresponding to the cp-map $\E_A$  are the $\gamma p$-roots of unit. Consequently, if we block $\gamma p$ spins, then we can write the new matrices $A_i$ as block-diagonal, with each block being injective and different (see Lemma 5 in Appendix A). We have now that the state $|\psi\>$ can be written as linear combination (with equal coefficients) of $\gamma p$ different injective MPS, each of them being a translation of each other. In Lemma 3 of Appendix A, we show that different injective MPS are orthogonal in the thermodynamic limit. Let $L = k \gamma p$ ($k \in \mathbb{N}$), using Jensen's inequality we have that $S(\rho_L) \geq - \log(\text{tr} (\rho_L^2) )$ which, by Lemma 4 in Appendix A implies $S(\rho_L) \geq \log(\gamma p) \geq \log (p),$ up to an exponentially small correction in $N-L$ and $k,$ as in the example proposed above.  

If the MPS has many blocks in its canonical form, we will show that one can treat each of these blocks as in the single block case, obtaining an extension of the last result. Lemma 8 gives us $\gamma \in \mathbb{N}$ such that all the blocks of the canonical form of $|\psi\>$ have period $\gamma p$. Let $L=k\gamma p,$ where $k\in \mathbb{N}.$ We observe that the reduced density operator of a region comprising $L$ sites verifies 
 \begin{equation}
 \label{eq:decomp}
 \rho_L=\oplus_{i=1}^n \mu_i \rho_i,
 \end{equation}
up to a correction exponentially small in $N-L$ and $k$ (see Lemma 4 in Appendix A).  The $\rho_i$'s are the reduced density matrices corresponding to single blocks, where repeated blocks are simply reflected in the $\mu_i$'s. Using the single block case, we can ensure again that $S(\rho_i) \geq \log (p)$ for all $i.$ It is clear, from (\ref{eq:decomp}) and the subaditivity of the entropy, that $S(\rho_L) \geq \log (p)$ up to another exponentially small correction in $N-L$ and $k,$ yielding the desired result.
\end{proof}

To prove the crucial Lemma 8, it will be enough to consider the characterization of symmetries for injective MPS \cite{SWPC09}, as well as an extension of the Lieb-Schutz-Mattis theorem for $U(1)$ symmetry which is explained in Lemma 6. The first \cite{SWPC09} will allow us to assert that all injective MPS corresponding to blocks must have the same symmetry and the same magnetization $m$. The later, that all these blocks should have a period multiple of $p.$ Moreover, Lemma 5 in Appendix A ensures that states corresponding to different blocks are necessarily different.

\section*{Large interaction length implies large entanglement} 
Now, we turn to the other situation where one can prove the appearance of entanglement. For that, let us consider again a translationally invariant MPS, $|\psi_A\>$, which is not the ground state of any short-range (gapped and frustration-free) Hamiltonian. Furthermore, let us assume that it is also far away (as specified below) from any other state with this property for any given interaction length. We will show that, as a consequence, its entanglement entropy will be large and, indeed, will scale with the range of the interaction.

In fact, if we denote by $\rho_A^L$ the reduced density operator of $|\psi\>$ for a (connected) region containing $L$ spins, we can prove the following Theorem

\begin{thm}\label{thm:interaction}
Let $|\psi_A\>$ be an MPS such that every state \,$|\tilde{\psi}\>$ which is the unique ground state of a gapped frustration-free Hamiltonian with interaction length $L$ verifies $\|\rho_A^L - \tilde{\rho}^L \|_1 \geq \epsilon$. Then, for sufficiently large regions $R,$ we have that the $\alpha$-Renyi entropy $S_{\alpha}(\rho_A^R) \ge a L + b \log \epsilon+c$, for $\alpha\leq\frac{1}{6}$ and where $a,b,c$ are constants depending only on the local physical dimension $d$ of  $|\psi_A\>.$
\end{thm}

This claim can be proved by contradiction. We will suppose that for every connected region and for $\alpha \leq \frac{1}{6}$ we have that the $\alpha$-Renyi entropy is upper bounded by an expression of the form $a L + b \log \epsilon+c$, for $\alpha\leq\frac{1}{6},$ where $a,b,c$ are constants depending on the physical dimension of  $|\psi_A\>.$ It will be enough to prove that this implies the existence of a state, the unique ground state of a gapped frustration-free Hamiltonian with interaction length $L,$ such that $\|\rho_A^L - \tilde{\rho}^L \|_1 < \epsilon.$

The hypothesis on $S_{\alpha}$ being small implies that we can find another MPS with a sufficiently small bond dimension, $\tilde{D}$ (in particular, $\tilde{D}\le d^{(L-1)/2}$) that is close enough to the original one. In order to do this, we will rely on \cite[Lemma 2]{VC} and on a new bound for reducing the bond dimension of an MPS. More precisely, this bound will be of the form $$\| \rho_A^L - \rho_{\tilde{A}}^L\|_1 \le 2 \sqrt{2}d^{L/2} \sqrt{L}\delta^{1/4}+(2L+3)\delta,$$ where $\rho_{\tilde{A}}^L$ will be the reduced density matrix which can be constructed from $\rho_A^L$ by substituting the Kraus operators $A_i$ by $PA_iP$ and $\Lambda = P \Lambda P,$  where $P = \sum_{i_1}^{\tilde{D}} \vert i \rangle \langle i \vert.$ It will be explained in further detail in Appendix B. 

Now, arbitrarily close to the MPS associated to the reduced density matrix $\rho^L_{\tilde{A}},$ there exists another which is the unique ground state of a Hamiltonian with gap and interaction length $L.$ Taking into account that the interaction length is closely related to the bond dimension at which the MPS reaches injectivity, this will be deduced from proving that all MPS (except for a set with measure zero) reach injectivity fast enough.  Standard Algebraic Geometry, as explained in Lemma 11 of Appendix C and \cite{ag, ag2}, reduces this problem to finding the existence of a single MPS displaying this property. The existence of such an example can be obtained using quantum expanders, as explained also in Appendix C.  

A more detailed proof can be given as follows:

\begin{proof}
Let us call $\lambda_i$ the ordered eigenvalues of $\rho_A^R,$ which can be taken as close as wanted to those of  $\Lambda\otimes \Lambda$ by enlarging region $R$  \cite[Lemma 2]{VC}. In this case, it is not difficult to see that, if we call $\mu_i$ the ordered elements of $\Lambda$, then $\sum_{i=\tilde{D}+1}^\infty \mu_i\le\sum_{i=\tilde{D}+1}^\infty \lambda_i= : \delta$. 

Suppose that, for $\alpha=\frac{1}{6}$ and for all $R,$ we can upper-bound the $\alpha$-Renyi entropy by
\begin{equation}\label{entropy}S_{\alpha}(\rho_A^R) \leq \frac{4}{5}\log\epsilon+ \frac{1}{10}(L \log d-\log L)-\log\frac{d}{4}\end{equation}

In Appendix C, we show that we can always construct a state, that we shall call $|\tilde{\psi}\>$, of the form

\begin{equation}\label{newMPS}
\vert \tilde{\psi} \rangle = \sum_{\begin{subarray} {l} i_{1}, \ldots, i_L \\  i_{L+1}, \ldots, i_N\end{subarray}} \text{tr} \left( \tilde{A}_{i_1} \cdots \tilde{A}_{i_L} B_{i_{L+1}}  C_{i_{L+2}} \cdots  C_{i_N} \right) \vert i_1 \cdots i_N \rangle
\end{equation} where $\tilde{A}_i, B_j, C_k \in \mathcal{M}_{\tilde{D} \times \tilde{D}},$  $\tilde{A}_i = PA_iP$ (being $A_i$ the Kraus operators defining the original MPS), with bond dimension $\tilde{D}\le d^{(L-1)/2}$ and such that the fixed point for the associated channel is $\tilde{\Lambda} = P \Lambda P,$ where we are considering that $P = \sum_i^{\tilde{D}} \vert i \rangle \langle i \vert.$ 

In Appendix C, we also prove that all states of this form (except a set of measure zero) reach injectivity in $L-1$ sites. Therefore, the one we have constructed in (\ref{newMPS}) is the unique ground state of a frustration-free Hamiltonian with interaction length $L$ \cite{PVWC07, FNW92}. Using a straighforward adaptation of \cite[Section 6]{FNW92}, this Hamiltonian is also gapped. Even though this state is not exactly translational invariant, it verifies that its normalized reduced density matrix for particles $1\ldots L$ is of the form

\begin{eqnarray*}
\rho_{\tilde{A}}^L &=&\sum_{\begin{subarray}{l} i_1, \ldots, i_L \\ j_1, \ldots, j_L \end{subarray}} \left( \sum_{\alpha, \beta} \langle \alpha \vert \left[ \tilde{A}_{i_1} \ldots \tilde{A}_{i_L} \tilde{\Lambda} \tilde{A}^{\dagger}_{j_L} \ldots \tilde{A}^{\dagger}_{j_1} \right] \vert \beta \rangle \right)
\end{eqnarray*} up to a exponentially small correction (see Appendix D).

This will allow us to use a bound, which is proved in Appendix B, which states that $$\| \rho_A^L - \rho_{\tilde{A}}^L\|_1 \le 2 \sqrt{2}d^{L/2} \sqrt{L}\delta^{1/4}+(2L+3)\delta$$$$\le 4 \sqrt{2}d^{L/2} \sqrt{L}\delta^{1/4} =:\epsilon',$$ since the first term in the sum is clearly larger than the second. It only remains to show that $\epsilon'\le \epsilon$, or equivalently, that $\delta\le \frac{\epsilon^4}{2^{10}d^{2L}\sqrt{L}}$. Since we have taken $R$ large enough, then up to a exponentially small correction in $R,$ we can state that $$\log(\delta)\le \frac{1-\alpha}{\alpha}\left(S_\alpha(\rho_A^R)-\log\frac{\tilde{D}}{1-\alpha}\right)\; .$$ 

Using this and the fact that that $\tilde{D}\ge d^{(L-2)/2}$, it is enough to prove
\begin{align*}
S_\alpha (\rho_A^R)&\le \frac{4\alpha}{1-\alpha} \log \epsilon +\frac{\log d}{2}\left(1-\frac{4\alpha}{1-\alpha}\right)L \\&-\frac{\alpha}{(1-\alpha)}(10+\frac{1}{2}\log L) -\log(1-\alpha)-\log d\\ &=
\frac{4}{5}\log\epsilon + \frac{1}{10}(L\log d-\log L)-\log\frac{5}{6}-\log \frac{d}{4}
\end{align*}
where, in the last step, we have set $\alpha=\frac{1}{6}.$ This is given by hypothesis in Eq. (\ref{entropy})\footnote{Note that the $\alpha$-Renyi is monotonically increasing in $\alpha$ }. Therefore, there exists a state $|\tilde{\psi}\>$, which is the unique ground state of a gapped frustration-free Hamiltonian with interaction length $L,$ such that $\|\rho_A^L - \tilde{\rho}^L \|_1 < \epsilon,$ as we wanted to prove.
\end{proof}

\section*{Conclusion} In this work, we have shown how MPS are powerful enough to provide formal proofs of certain believed statements on strongly correlated spin systems that were lacking a mathematical treatment. In particular, we have stated and proved that, for the state of a quantum spin chain, either a large fractionalization in the magnetization or the impossibility of being well approximated by the ground state of a local Hamiltonian demands large entanglement. Moreover, since MPS seem to be the right representation for the low energy sector of 1D systems, one may postulate the results being true in full generality.

\section*{Acknowledgments}

We acknowledge discussions with R. Or\'us, M.C. Banuls, G. Sierra and specially  J.J. Garc\'ia-Ripoll.  DPG and IC acknowledge the support and hospitality of Perimeter Institute, where
some part of this work was carried out. This work was supported by the European projects QUEVADIS, CHIST-ERA CQC and Spanish grants QUITEMAD and MTM2011-26912.

\appendix
\section{Technical lemmas for the proof of Theorem 1}
Our first aim is to state and prove a couple of lemmas formalizing the claim: ``for injective MPS, different means orthogonal''.

\begin{lem}\label{lem:orthogonal}
Given two injective MPS, $|\psi_A\>$ and $|\psi_B\>,$ then $\||\psi_A\>\|,\||\psi_B\>\|=1$ up to an exponentially (in $N$) small correction. Moreover, either both are equal for all $N$, or $\text{tr}_{N-L} \vert \psi_A \> \< \psi_B | =0$ up to an exponentially (in $N-L$) small correction. In particular, $|\<\psi_A|\psi_B\>|=0$ up to an exponentially (in $N$) small correction.
\end{lem}

\begin{proof}
It is easy to see that  $\<\psi_B|\psi_A\>=\tr{\mathcal{E}_{A,B}^N}$, where $\mathcal{E}_{A,B}=\sum_{i}{A}_i \otimes \bar{B}_i$. Moreover, it is clear that the eigenvalues of $\mathcal{E}_{A,B}$ are the same as those of the map $\E_{A,B}(X)=\sum_iA_i XB_i^\dagger$, which gives $\||\psi_A\>\|,\||\psi_B\>\|=1$ up to a exponentially small correction. To finish, it is enough to see that all eigenvalues $\lambda$ of $\E_{A,B}$ verify that $|\lambda|<1$. We will use the conditions verified by the canonical form of an injective MPS, that is: (i) $\sum_{i}A_i A_i^\dagger=\1$, (ii) $\sum_i A_i^\dagger \Lambda_A A_i=\Lambda_A$ for a diagonal positive full rank matrix $\Lambda_A,$ and (iii) the  cp map $\E_A$ defined as 
 \begin{equation}
 \E_A(X)=\sum_i A_i XA_i^\dagger 
 \end{equation}
has 1 as its unique non-degenerate eigenvalue of maximal modulus.

Let us take $X$ such that $\sum_{i}A_i XB_i^\dagger=\lambda X$, using (i) for $A$ and (ii) for $B$  we get
\bea
 && |\lambda| |\tr{X\Lambda_B X^\dagger}|=\left|\sum_i \tr{ A_iX B_i^\dagger \Lambda_B X^\dagger}\right| \nonumber\\
 &\le& \left[\sum_i \tr{XB_i^{\dagger}\Lambda_B B_i X^\dagger}\right]^{1/2}
 \left[\sum_i \tr{A_i^{\dagger}X\Lambda_B X^\dagger A_i}\right]^{1/2} \nonumber\\
 &=&|\tr{X\Lambda_B X^\dagger}|,
 \label{ineq}
 \eea
 where we have used the Cauchy-Schwarz inequality and $\tr{X\Lambda_B X^\dagger}>0$. So, if $|\lambda|\ge 1,$ we must have an equality and, therefore, $\alpha\Lambda_B^{1/2} X^\dagger A_i =\Lambda_B^{1/2}
B_i X^\dagger$.  Multiplying by the adjoint expression, summing in $i$, taking traces and using (i) and (ii) again we get that $|\alpha|=1$ and, hence, $\alpha=e^{i\theta}$. Finally, since $\Lambda_B$ is invertible, we get $\sum_iB_iX^\dagger XB_i^\dagger=X^\dagger X$, which, by (iii), leads to $X^\dagger X=\1$ and implies that $e^{i\theta}A_i=XB_iX^\dagger$. This means that $|\psi_A\>$ and $|\psi_B\>$ are equal up to a global phase, for all $N$.
\end{proof}

A similar proof gives the following

\begin{lem}\label{lem:density}
Given an MPS of the form $|\psi\>=\sum_{r=1}^n \lambda_r |\psi_r\>$ such that the $|\psi_r\>$ are different injective MPS, then $\tr{\rho^L_r \rho^L_s} \propto \delta_{rs} + O(e^{-L}) + O(e^{-(N-L)}),$ being $\rho^L_r$ the reduced density matrix for $L$ particles associated to $|\psi_r\>.$ 
\end{lem}

The next thing we need is the following modification of \cite[Theorem 5]{PVWC07}.

\begin{lem}\label{thm:periodic}
Consider any MPS $|\psi_A\>\in\mathbb{C}^{d\otimes N}$ which has
only one block in its canonical form with $D\times D$ matrices
$\{A_i\}$ and such that $\E_A$ has $\beta$ eigenvalues of
modulus one. If $\beta$ is a factor of $N$, then the state can be
written as a superposition of $\beta$ $\beta$-periodic \emph{different and injective} MPS with equal coefficients and bonds $D_i$ (also with the property that $\sum_i D_i=D$). Otherwise, if $\beta$ is not a factor of $N$, then $|\psi_A\>=0$.
\end{lem}

\begin{proof} The only thing to prove is that the $\beta$ $\beta$-periodic states are injective and different. In the proof of \cite[Theorem 5]{PVWC07}, based on \cite{FNW92}, one proves the existence of a set of orthogonal projectors $\{P_k\}$ with
$\sum_kP_k=\mathbbm{1}$ such that
\begin{equation}\label{block-diag}
\E_A^{\beta}(X)=\sum_{j,k}P_j\E_A^{\beta}(P_jXP_k)P_k,\end{equation}
 $\E_A^{\beta}$ has $1$ with degeneracy exactly $\beta$ as the only eigenvalue of modulus 1, and each block in the block-diagonal form of the Kraus operators of $\E_A^{\beta}$ given by (\ref{block-diag}) corresponds to one of the $\beta$-periodic states. Moreover, the space of fixed points is generated by $P_k$ and the space of fixed points of the adjoint map is generated by $P_k\Lambda P_k$.

 The cp maps associated to the $\beta$-periodic states are then $\E_k(X)=P_k\E^{\beta}(P_kXP_k)P_k$ (restricted to inputs with $X=P_kXP_k$). It is clear that $P_k$ is its only fixed point, $P_k\Lambda P_k$ the only fixed point of the adjoint map, and there is no other eigenvalue of modulus 1, which shows that all $\beta$-periodic states are injective. Now, if two of them were equal, we would reach a contradiction in the following way. For simplicity, we reason in the case of $2$ $2$-periodic states but the argument can be adapted straightforwardly to the general case. $\E_A^2$ has block-diagonal Kraus operators of the form $B_i\otimes |0\>\<0|+ C_i\otimes |1\>\<1|$. By the hypotheses and the uniqueness of the canonical form for injective MPS, $B_i=e^{i\theta}UC_iU^\dagger$ for all $i$. Then, apart from $\1\otimes |0\>\<0|$ and $\1\otimes |1\>\<1|$,  we also get $U\otimes |0\>\<1|$ as an eigenvector of $\E_A^{\beta}$ with eigenvalue of modulus 1; the desired contradiction.
\end{proof}

Finally, we need the following version of the Lieb-Schultz-Mattis theorem for $U(1)$ symmetry. It is interesting to note that it does not use any MPS structure, so it is valid in full generality and for any spatial dimension.
Let us recall that, in \cite[Lemma 17]{NJP}, we showed that any quantum state with a $U(1)$ symmetry given by the canonical generator of spin $S_z^{(J)}$ verifies that \begin{equation}\label{eq:sym}u_g^{\otimes N}|\psi\>=e^{igNm}|\psi\>\end{equation} with $u_g=e^{igS_z^{(J)}}$ and a  magnetization per particle $m$.
\begin{lem}\label{LSM}
Let $m$ be any rational number and  $p\in \mathbb{N}$ such that there exist two quantum states of (local spin $J$ and) $pN$ and $(N+1)p$ particles respectively, for some $N,$ having both of them magnetization per particle $m$. Then $p(J-m)=q$ with $q$ integer.
\end{lem}
\begin{proof}
By expanding equation (\ref{eq:sym}) in the canonical basis, we get $\sum_{k_1\cdots k_{pN}}c_{k_1\cdots k_{pN}}e^{ig\sum_j k_j}|k_1\cdots k_{pN}\>=\sum_{k_1\cdots k_{pN}}e^{igpNm}c_{k_1\cdots k_{pN}}|k_1\cdots k_{pN}\>$. Since it is a basis and the state is not zero, there must exist $k_1,\cdots, k_{pN}\in  \{-J,-J+1,\ldots J-1,J\}$ such that $\sum_j k_j=Npm$. For the same reason, there must exist  $k_1',\cdots, k_{pN+p}'\in  \{-J,-J+1,\ldots J-1,J\}$ such that $\sum_j k_j'=(Np+p)m$. Subtracting, we get that $mp=\sum_j k_j'-\sum_j k_j$ has the same character (integer or semi-integer) as $pJ$.
\end{proof}

With this at hand, if we consider an MPS $|\psi\>$ of spin $J$ and $pN$ particles with a $U(1)$ symmetry, given by the canonical generator of spin $S_z^{(J)}$, we have the following lemma.
\begin{lem}\label{thm:GLSM}
Let $p$ be the smallest integer such that, after blocking $p$ sites together, $|\psi\>$ has a block-diagonal representation with injective blocks. Then $p (J - m)= q$, with $q$ an integer.
\end{lem}

To see it we consider blocks of $p$-sites. From \cite[Theorem 5]{SWPC09}, we know that each block is an injective MPS with the same symmetry. Since, by Lemma \ref{lem:orthogonal}, states corresponding to different blocks are equal or linearly independent, all of them must have also magnetization $m$. Now, by the characterization of symmetries for injective MPS \cite{SWPC09}, we know that the matrices defining each block inherit the symmetry and therefore the associated MPS has magnetization $m$ for all system sizes that are multiple of $p$. Lemma \ref{LSM} finishes the argument.

We also get a reciprocal. 

\begin{lem}\label{thm:GLSM-reverse}
Let us assume that $J-m=\frac{q}{p}$ with ${\rm gcd}(p,q)=1$ in a $U(1)$ symmetric MPS, then there exists $\gamma \in \mathbb{N}$ such that
the MPS has only $\gamma p$-periodic blocks. Moreover (trivially from Lemma \ref{thm:periodic}), states belonging to blocks of different periods are different.
\end{lem}
\begin{proof}
As above, all injective MPS corresponding to the blocks must have the same symmetry and the same magnetization $m$. Therefore, Lemma \ref{LSM} shows that only blocks of period multiple of $p$ can appear.
\end{proof}

We are finally ready to prove Theorem 1:

\begin{proof}
Let $|\psi \rangle$ be an MPS, which is translational and $U(1)$ invariant. We also impose that this MPS has spin $J$ and magnetization per particle $m,$ verifying $J-m = \frac{q}{p}$ ($p$ and $q$ coprime) and consider its canonical form. Lemma \ref{thm:GLSM-reverse} gives us $\gamma \in \mathbb{N}$ such that all the blocks of  the canonical form of $|\psi\>$ have period $\gamma p.$ Consequently, if we block $\gamma p$ spins, then we can write the new matrices $A_i$ as block-diagonal, with each block being injective and different (see Lemma \ref{thm:periodic}). Using Lemma 3, the injective and different MPS associated to the blocks are also orthogonal in the thermodynamic limit. 

Let $L = k \gamma p,$ where $k \in \mathbb{N}.$ We observe that the reduced density matrix of size $L,$ verifies 
 \begin{equation}
 \label{eq:formsum}
 \rho_L=\sum_{i=1}^n  \mu_i \rho_i,
 \end{equation}
up to a correction exponentially small in $N-L,$ where we have used Lemma 3. Here, the $\rho_i$'s are the reduced density matrices corresponding to single blocks (giving rise to different states) and repeated blocks in this sum are simply reflected in the $\mu_i$'s. 

Analizing the single block case, we can ensure that $S(\rho_i) \geq \log(\gamma p) \geq \log (p)$ for all $i,$ up to an exponentially small correction in $N-L$ and in $k.$ This is deduced by using Jensen's inequality so that $S(\rho_i) \geq - \log( \text{tr}(\rho_i^2) ),$ and recalling Lemma \ref{lem:density}. Using (\ref{eq:formsum}) and the concavity of the Von Neumann entropy, it is clear that, if we have several blocks then $S(\rho_L) \geq \min_i S(\rho_i) \geq \log (p)$ up to another exponentially small correction in $N-L$ and $k,$ yielding the desired result.
\end{proof}

\section{Bounds on MPS approximation}

Let ${A_i \in \mathcal{M}_D}$ be the canonical Kraus operators
defining an injective MPS, with $\Lambda$ as its fixed point. We define the
normalized reduced density matrix for $L$ particles $\rho_A^L$, up to a correction exponentially small in $N-L$, by
\begin{equation}
\rho_A^L = \sum_{\begin{subarray}{l} i_1,\ldots , i_L \\
j_1,\ldots, j_L
\end{subarray}}
\tr{A_{j_L}^{\dagger}\cdots A_{j_1}^{\dagger}\Lambda A_{i_1}\cdots
A_{i_L}} |i_1 \cdots i_L\>\<j_1 \cdots j_L| 
\end{equation} 

We will also define $\rho_{\tilde{A}}^L$ as the normalized density matrix resulted of projecting the Kraus
operators (and the fixed point) into a subspace of dimension
$\tilde{D} \leq D$, that is, $\tilde{A}_i = P A_i P$ and
$\tilde{\Lambda} = P \Lambda P$ with
$P=\sum_{i=1}^{\tilde{D}}|i\>\<i|.$ $\E$ will be the cp map associated to $A_i$ and $\tilde{\E}$ the one
associated to $\tilde{A}.$  Taking all this into account, we can
state and prove the following Theorem:

\begin{thm*}
$$\| \rho_A^L - \rho_{\tilde{A}}^L \|_2 \leq 2 \tr{\tilde{\Lambda}^{1/2}}\sqrt{L}\delta^{1/4} + (2L+3)\delta,$$ $$\| \rho_A^L - \rho_{\tilde{A}}^L \|_1 \leq 2\sqrt{2}\tilde{D}\sqrt{L}\delta^{1/4}+  (2L+3)\delta$$
where $\delta=\tr{\Lambda-\tilde{\Lambda}}$.
\end{thm*}

In order to do this, we must prove the following two Lemmas as preliminary results:

\begin{lem} \label{lemma1}
$\| \tilde{\E}^L(\Lambda) -\Lambda\|_1 \leq 2L\delta$. In particular,
$\tr{\tilde{\E}^L(\Lambda)} \geq 1 - 2 L \delta$.
\end{lem}
\begin{proof}
Using both the definition of $\delta$ and that $\E$ is contractible for
the 1-norm, we get that $\| \Lambda -
\E(P\Lambda P) \|_1 \leq \delta.$ The map $P \bullet P$ is also contractible for the
1-norm, so
\begin{eqnarray*}
&& \|\Lambda - P\E(P \Lambda P) P\|_1 \\
&\leq& \| \Lambda - P\Lambda P\|_1
 +\|P\Lambda P - P \E(P\Lambda P)P\|_1 \\
&\leq& 2\delta
\end{eqnarray*}
This means that $\| \Lambda - \tilde{\E}(\Lambda)\|_1 \leq 2\delta,$ since $\tilde{\E}(\Lambda) = P\E(P \Lambda P) P.$
However, $\tilde{\E}$ is also contractible respect to the 1-norm, so
\begin{eqnarray*}
\| \Lambda - \tilde{\E}^2(\Lambda) \|_1 &\leq& \|\Lambda -
\tilde{\E}(\Lambda)\|_1  +\| \tilde{\E}(\Lambda) -
\tilde{\E}^2(\Lambda)\|_1 \\ &\leq& 4\delta
\end{eqnarray*}
The result can be obtained by induction.
\end{proof}

We will now define, under the previous notation for the Kraus operators and the fixed point, the following operators
\begin{widetext}
\begin{equation*}
\sigma_{A}  = \sum_{\begin{subarray}{l}
i_1,\ldots , i_L\\ j_1,\ldots, j_L
\end{subarray}}
\tr{A_{j_L}^{\dagger}\cdots A_{j_1}^{\dagger}\tilde{\Lambda}
A_{i_1}\cdots  A_{i_L}} |i_1 \cdots i_L\>\<j_1 \cdots j_L |
\end{equation*}
\begin{eqnarray*}
&&\sigma_{A,P}  = \sum_{\begin{subarray}{l}
i_1,\ldots , i_L\\ j_1,\ldots, j_L
\end{subarray}}
\tr{P A_{j_L}^{\dagger}\cdots A_{j_1}^{\dagger}\tilde{\Lambda}
A_{i_1}\cdots  A_{i_L} P} |i_1 \cdots i_L\>\<j_1 \cdots j_L |
\end{eqnarray*} \end{widetext}where it is important to note that $\sigma_{A,P}$ is a positive operator.

\begin{lem} \label{lemma2}
$\| \rho_A - \rho_{\tilde{A}}\|_2 \leq
\|\sigma_{A,P} -
\phi_{\tilde{A}}\|_2 + (2L+3)\delta$, where
$\phi_{\tilde{A}} = \tr{\tilde{\E}^L(\Lambda)}
\rho_{\tilde{A}}$ is the not normalized reduced density
matrix generated by $\tilde{A}_i.$ The same holds changing the 2-norm by the 1-norm in both sides of the inequality.
\end{lem}

\begin{proof}
By using the triangular inequality and the fact that $\|\cdot\|_2\le \|\cdot\|_1$,
\begin{eqnarray*}
\| \rho_A  - \rho_{\tilde{A}}\|_2  &\leq& \|
\rho_A - \sigma_A\|_1 +  \| \sigma_A -
\sigma_{A,P}\|_1 \\ &+& \|
\sigma_{A,P} - \phi_{\tilde{A}}
\|_2 + \| \phi_{\tilde{A}} -
\rho_{\tilde{A}}\|_1\; 
\end{eqnarray*}
The first term can be calculated exactly
\begin{eqnarray*}
\| \rho_A - \sigma_A\|_1 &=& \sum_{i_1 , \ldots,
i_L} \rm tr \Big [ A^{\dagger}_{i_L} \cdots A^{\dagger}_{i_1}
(\Lambda - \tilde{\Lambda}) A_{i_1} \cdots A_{i_L} \Big ] \\ &=&
\delta\; .
\end{eqnarray*}
The first equality holds because the operator is positive and the
1-norm can be replaced by a trace and the second one holds because $\E$ is trace preserving. The second term can be bounded in a similar
way.
\begin{eqnarray*}
\| \sigma_A - \sigma_{A,P}\|_1 &=&
\tr{P^{\perp} \sum_{i_1, \ldots , i_L} A^{\dagger}_{i_L} \cdots
A^{\dagger}_{i_1} \tilde{\Lambda} A_{i_1} \cdots A_{i_L} P^{\perp}} \\
&\leq& \delta + \tr{P^{\perp} \sum_{i_2, \ldots , i_L}
A^{\dagger}_{i_L} \cdots A^{\dagger}_{i_2} \Lambda A_{i_2}\cdots
A_{i_L}}
\end{eqnarray*}
This holds because $\| \Lambda - \E(\tilde{\Lambda})\|_1 = \|
\E(\Lambda - \tilde{\Lambda}) \|_1 = \delta$, since $\E$ is trace
preserving and $\E(\Lambda) = \Lambda$. Therefore, $\|
\sigma_A - \sigma_{A,P}\|_1 \leq \delta +
\tr{P^{\perp} \Lambda} = 2 \delta$.

Finally, the last term can be bounded using Lemma \ref{lemma1}
because
\begin{equation*}
\| \phi_{\tilde{A}}  -
\rho_{\tilde{A}}\|_1 = -1 +
\tr{\tilde{\E}^L(\Lambda)} \leq 2\delta L
\end{equation*}
We obtain the result by collecting all bounds above.
\end{proof}

Now 
\begin{eqnarray*}
&& \|\sigma_{A,P}- \phi_{\tilde{A}} \|_2^2 \\
&\le& \Big [ \Big ( \tr{Q (\EE^{*})^L (\tilde{\Lambda} \otimes
\tilde{\Lambda}) \EE^L Q}  - \tr{Q (\FF^{*})^L (\tilde{\Lambda}
\otimes \tilde{\Lambda}) \FF^L Q} \Big ) \\
&+& \Big ( \tr{Q (\FF^{*})^L (\tilde{\Lambda} \otimes
\tilde{\Lambda}) \FF^L Q} - \tr{Q
(\tilde{\EE}^{*})^L (\tilde{\Lambda}\otimes \tilde{\Lambda})
\tilde{\EE}^L Q} \Big ) \Big ]
\end{eqnarray*} where $\EE=\sum_iA_i\otimes \bar{A_i}$,  $Q = P\otimes P$ and $\FF = (\1 \otimes P ) \EE
(\1 \otimes P)$.

\qquad \\
We have now all the necessary tools to prove the main Theorem:

\begin{proof}[Proof of the Theorem]
We start by bounding the term $\mu = \Big | \tr{Q (\EE^{*})^L
(\tilde{\Lambda} \otimes \tilde{\Lambda}) \EE^L Q}  - \tr{Q
(\FF^{*})^L (\tilde{\Lambda} \otimes \tilde{\Lambda}) \FF^L Q}
\Big |$. This can be done by adding and subtracting terms such that
they differ in one projector, i.e.
\begin{eqnarray*}
\mu &\leq& \sum_{r=1}^{L-1} \Big | \tr{\FF^L Q (\FF^{*})^{r-1}\EE^* (\1 \otimes
P^{\perp}) (\EE^{*})^{L-r}(\tilde{\Lambda} \otimes
\tilde{\Lambda})} \Big | \\ &+& \sum_{s=1}^{L-1} \Big | \tr{\EE^s (\1 \otimes
P^{\perp})\EE \FF^{L-s-1} Q (\EE^{*})^L (\tilde{\Lambda} \otimes
\tilde{\Lambda})} \Big | \\ &=& \sum_r \mu_r + \sum_s \nu_s
\end{eqnarray*}
 Let us bound the first family of terms. By applying the Schwarz inequality  $|\tr{ \sum_i A_i B_i }|\leq \left|\tr{\sum_i A^{\dagger}_i A_i
}\right|^{\frac{1}{2}} \left|\tr{\sum_i B_i B^{\dagger}_i }\right|^{\frac{1}{2}}$,
\begin{widetext}
{\small
\begin{equation*}
\mu_r  = \left|{\rm tr}\left[  \sum_{\begin{subarray}{l} k_1,\ldots, k_L\\i_1,\ldots , i_r\\ j_1,\ldots, j_{L-r} \end{subarray}} \left(\sqrt{\tilde{\Lambda}} A_{k_1}\cdots A_{k_L}PA^{\dagger}_{i_1} \cdots A^{\dagger}_{i_r} A^{\dagger}_{j_1} \cdots A^{\dagger}_{j_{L-r}} \tilde{\Lambda}^{1/4} \otimes  \tilde{\Lambda}^{1/4}\right) \right.\right.\cdot\end{equation*}
\begin{equation*}\left.\left.\left(\overline{ \tilde{\Lambda}^{1/4} \otimes \tilde{\Lambda}^{1/4} \tilde{A}_{k_1}\cdots  \tilde{A}_{k_L}\tilde{A}^{\dagger}_{i_1} \cdots \tilde{A}^{\dagger}_{i_{r-1}} A_{i_r}^\dagger P^{\perp} A^{\dagger}_{j_1} \cdots A^{\dagger}_{j_{L-r}} \sqrt{\tilde{\Lambda}}}\right) \right]\right|
\end{equation*}
\begin{equation*}
\leq \left|\tr{\sum_{\begin{subarray}{l} k_1,\ldots, k_L\\i_1,\ldots , i_r\\
j_1,\ldots, j_{L-r} \end{subarray}} \tilde{\Lambda}^{1/4} A_{j_{L-r}}\cdots A_{j_1}
A_{i_r} \cdots A_{i_1} P A^{\dagger}_{k_L}\cdots A^{\dagger}_{k_1} \tilde{\Lambda} A_{k_1}\cdots A_{k_L} PA^{\dagger}_{i_1} \cdots
A^{\dagger}_{i_r} A^{\dagger}_{j_1} \cdots A^{\dagger}_{j_{L-r}}\tilde{\Lambda}^{1/4}
\otimes \tilde{\Lambda}^{1/2} }\right|^{\frac{1}{2}} \cdot
\end{equation*}
\begin{equation*}
\cdot
 \left|\tr{\sum_{\begin{subarray}{l}k_1,\ldots, k_L\\
i_1,\ldots , i_r\\ j_1,\ldots, j_{L-r} \end{subarray} }\tilde{\Lambda}^{1/2} \otimes \tilde{\Lambda}^{1/4} \tilde{A}_{k_1}\cdots  \tilde{A}_{k_L}  \tilde{A}^{\dagger}_{i_1} \cdots
\tilde{A}^{\dagger}_{i_{r-1}}A_{i_r}^\dagger P^{\perp}
A^{\dagger}_{j_1} \cdots
A^{\dagger}_{j_{L-r}}\tilde{\Lambda}A_{j_{L-r}} \cdots A_{j_1}
P^{\perp}A_{i_r}\tilde{A}_{i_{r-1}} \cdots \tilde{A}_{i_1}\tilde{A}^{\dagger}_{k_L}\cdots \tilde{A}^{\dagger}_{k_1}\tilde{\Lambda}^{1/4}
}\right|^{\frac{1}{2}}\; .
\end{equation*}}
\end{widetext}
 The
first term is equal to 
\begin{eqnarray*}&& \tr{\tilde{\Lambda}^{1/2}}^{1/2}\tr{P\E^L(\tilde{\Lambda})P \tilde{\E}^L(\tilde{\Lambda}^{1/2})}^{1/2} \\ &\le& \tr{\tilde{\Lambda}^{1/2}}. 
\end{eqnarray*} 

The second term is equal to
\begin{eqnarray*}
&& \tr{\tilde{\Lambda}^{1/2}}^{1/2}\tr{\tilde{\E}^{r-1}\circ\E\left(P^\perp \E^{L-r}(\tilde{\Lambda})P^\perp\right) \tilde{\E}^L(\tilde{\Lambda}^{1/2})}^{1/2}\\ &\le& \delta^{1/2}\tr{\tilde{\Lambda}^{1/2}}\,
\end{eqnarray*}
where we have used that $\tilde{\Lambda} \leq \Lambda$ (hence,  $\tr{P^\perp \E^{L-r}(\tilde{\Lambda})P^\perp}\le \delta$), and that both $\E$ and $\tilde{\E}$ are contractible for the trace norm. Therefore,
$\mu_r \le \tr{\tilde{\Lambda}^{1/2}}^2\sqrt{\delta}$. The result for the $\nu_s$
is exactly the same, so it
follows that $\mu \leq 2 L\tr{\tilde{\Lambda}^{1/2}}^2\sqrt{\delta}$.

The other term can be calculated in the same way, by replacing
$\EE \rightarrow \FF$ and $\FF \rightarrow \tilde{\EE}$, and it gives exactly the same estimate.

The second inequality follows from the first one, $\tr{\tilde{\Lambda}^{1/2}}\le \sqrt{\tilde{D}}\tr{\Lambda}$, and the fact that $\sigma_{A,P} - \phi_{\tilde{A}}$ has rank $\le 2\tilde{D}$, which then gives $$\| \sigma_{A,P} - \phi_{\tilde{A}}\|_1\le \sqrt{2\tilde{D}}\|\sigma_{A,P} - \phi_{\tilde{A}}\|_2.$$
\end{proof}

\section{Injectivity can be reached fast}

We will prove here the following technical lemma

\begin{lem}\label{random}
Every MPS (with the exception of a zero-measure set) of the form 
\begin{equation}
\vert \tilde{\psi} \rangle = \sum_{\begin{subarray} {l} i_{1}, \ldots, i_L \\  i_{L+1}, \ldots, i_N\end{subarray}} \text{tr} \left( A_{i_1} \cdots A_{i_L} B_{i_{L+1}}  C_{i_{L+2}} \cdots  C_{i_N} \right) \vert i_1 \cdots i_N \rangle
\end{equation} where $A_i, B_j, C_k \in \mathcal{M}_{D \times D}$ and $L\ge \frac{2\log D}{\log d}$ reaches injectivity in every region of length $L-1.$
\end{lem}
\begin{proof}
Since the set of MPS failing this property is clearly a projective algebraic subvariety of $\left( \C^D\otimes \C^D\otimes \C^d \right)^{\otimes 3}$, standard algebraic geometry tells us that, if this set is non-empty, since $\left( \C^D\otimes \C^D\otimes \C^d \right)^{\otimes 3}$ is irreducible then both projective varieties must be equal \footnote{It will be a straightforward conclusion after relating the following ideas concerning projective sets such that $X \subset Y$ (which always implies $\dim X \leq \dim Y$). On the one hand, if $X$ is a non-empty open subset of $Y,$ then $\dim X = \dim Y$. On the other hand, if $\dim X = \dim Y$ and $Y$ is irreducible then $X = Y.$ See Refs. \cite{ag, ag2} for more details about these results and related concepts in algebraic geometry}. Therefore, it is enough to find a single MPS reaching injectivity as stated in this lemma, which has been verified numerically up to $D=200$ and $d=50,$ and also analytically in the next Lemma of this Appendix using quantum expanders \footnote{Note that our analytical proof gives a slightly worse condition for the $L$ needed to reach injectivity (in terms of $D$ and $d$), but suffices, nevertheless, to prove the main Theorem.}.
\end{proof}

It is proven in \cite{expanders} that for all $d \geq 4$, there exists a Hermitian
trace-preserving completely positive map
$$\E(X)=\sum_{i=1}^d A^\dagger_iXA_i$$ such that $|\lambda_2|\le \left(\frac{2\sqrt{d-1}}{d}\right)\left(1+O\left(\log(D) D^{\frac{-2}{15}}\right)\right)$,
where $A_i\in M_D$.

Take the MPS $|\psi\>$ generated by the matrices $A_i$ and consider the map
$$\Gamma_n(X)=\sum_{i_1\cdots i_n}\tr{XA_{i_1}\cdots A_{i_n}}|i_1\cdots i_n\>$$
We want to show
 \begin{thm}
 Assuming $D$ is large enough, $\Gamma_n$ is an injective map for \footnote{Indeed, $k$ can be made arbitrarily close to $4$ at the price of enlarging $d$}  $$n\ge \left[\frac{k \log(D)} {\log(d)}\right]+1,\: K=8, d>16$$
 \end{thm}

This will be a consequence of the following

 \begin{lem}\label{lem}
$$ \sup_{\tr{X^\dagger X}=1} \left|\Gamma_n(X)^\dagger\Gamma_n(X)-\frac{1}{D}\tr{X^\dagger X}\right|\le D|\lambda_2|^n$$
 \end{lem}

\begin{proof}
	
	 Considering in $\mathcal{M}_D$  the usual Hilbert-Schmidt Hilbert structure, it is easy to see that the LHS is equal to
$$	 \left\| \Gamma_n^* \Gamma_n - \frac{1}{D} \1 \right\|_\text{op} $$ for the usual operator norm on the Hilbert space $\mathcal{M}_D$.

Moreover, in coordinates, calling $\mathcal{E}=\sum_{i}A_i\otimes \bar{A_i}$, we have that
$$ \Gamma_n^* \Gamma_n - \frac{1}{D} \1 =\sum_{abcd} \left(\<cd|\mathcal{E}^n|ab\> -\frac{1}{D}\delta_{ab}\delta_{cd}\right) |bd\>\<ac| \;.$$
just identifying $\mathcal{M}_D=\C^D\otimes \C^D$ and calling $|ij\>$ to the canonical (matrix) basis there.

Since for each operator on an $n$ dimensional Hilbert space, $\|\cdot\|_{op}\le \|\cdot\|_2\le\sqrt{n} \|\cdot\|_{op}$, being $\|\cdot\|_2$ the Hilbert-Schmidt norm, and using that the Hilbert-Schmidt norm is invariant under arbitrary rearrangements of the coordinates, we get that
\begin{eqnarray*}
&& \left\| \Gamma_n^* \Gamma_n - \frac{1}{D} \1 \right\|_\text{op} \\ &\le& D
\left\|\sum_{abcd} \left(\<cd|\mathcal{E}^n|ab\> -\frac{1}{D}\delta_{ab}\delta_{cd}\right) |ab\>\<cd|\right\|_{op} \\ &=& D
\|\mathcal{E}^n-\frac{1}{D}|\1\>\<\1|\|_{op}=D\|\mathcal{E}^n-\mathcal{E}^\infty\|_{op}\\ &=&D\|\E^n-\E^\infty\|_{op} =D|\lambda_2|^n
\end{eqnarray*} where we have used in the last step that $\E$ is hermitian and $|\1\>$ denotes the unnormalized vector $\sum_{i=1}^D |ii\>$.

\end{proof}

\begin{figure}
\includegraphics[width=8cm]{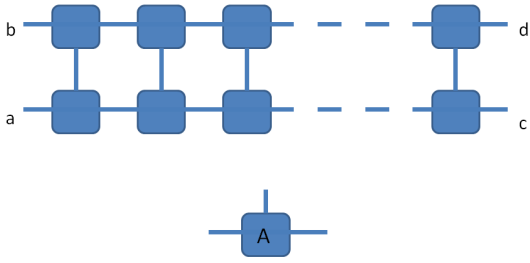}
\caption{Given the tensor $A=(\<\alpha|A_i|\beta\>)_{i\alpha\beta}$ which defines the MPS, and with the usual convention that rotating means complex conjugation, we can represent the map $\Gamma_n^*\Gamma_n$ as the map in the figure from systems $ac$ to systems $bd$ and the map $\E^n$ as the same figure but now mapping systems $cd$ to systems $ab$}\label{Fig1}
\end{figure}

\begin{proof}[Proof of Theorem 12]
$\Gamma_n$ must be injective as long as
\begin{equation}\label{eq:lambda2}|\lambda_2|^n<\frac{1}{D^2}.\end{equation} Otherwise, taking a (normalized) $X$ such that $\Gamma_n(X)=0$, we would get a contradiction to Lemma \ref{lem}.
	Since we know from \cite{expanders} that $|\lambda_2|\le \left(\frac{2\sqrt{d-1}}{d
	}\right)\left(1+O\left(\log(D) D^{\frac{-2}{15}}\right)\right)$ it suffices to take
	$n$ such that
		$$ \left(\frac{2\sqrt{d-1}}{d
	}\right)^{n}\left(1+O\left(\log(D) D^{\frac{-2}{15}}\right)\right)^{n}< \frac{1}{D^2}\; .$$
		Taking logarithms
	   $$ 2\log(D) + n \log \left[	\left({2\sqrt{d-1} \over d
	}\right)\left(1+O\left(\log(D) D^{{-2 \over 15}}\right)\right)  \right] < 0  $$
		which is equivalent to
		$$  n  > {2\log(D) \over \log \left[	\left({d \over 2\sqrt{d-1}
		}\right) \right] - \log \left(1+O\left(\log(D) D^{{-2 \over 15}}\right)\right) }  $$
It is clear that taking $D$ large enough we can upper-bound the RHS by
$$\left[{2\log(D) \over \log \left(	\left({d \over 2\sqrt{d-1}}\right) \right)}\right]+1$$
But now
\begin{align*}
{2\log(D) \over \log \left(	\left({d \over 2\sqrt{d-1}}\right) \right)} &= {4\log(D) \over 2\log(d) - \log (4) - \log(d-1)} \\
&\leq  \frac{4K \log(D)} {\log( d)}
\end{align*}
as long as $\frac{1}{K}\le 1-\frac{2}{\log d}$, which finishes the proof of the Theorem.
\end{proof}

\section{Some results for non translational invariant MPS}

\begin{lem}\label{new-state}
Let $A_i, \Lambda \in \mathcal{M}_{D}$, then there exist $B_i, C_i \in \mathcal{M}_{D}$ such that if we consider the state
\begin{equation}
\vert \psi \rangle = \sum_{\begin{subarray} {l} i_{1}, \ldots, i_L \\  i_{L+1}, \ldots, i_N\end{subarray}} \text{tr} \left( A_{i_1} \ldots A_{i_L} B_{i_{L+1}}  C_{i_{L+2}} \ldots  C_{i_N} \right) \vert i_1, \ldots i_N \rangle
\end{equation} then the normalized reduced density matrix for $L$ particles (particles 1-$L$) is
\begin{equation}
\rho_{1...L} = \sum_{\begin{subarray}{l} i_1, \ldots, i_L \\ j_1, \ldots, j_L \end{subarray}} \text{tr} \left( A^\dagger_{j_L} \ldots A^{\dagger}_{j_1} \Lambda A_{i_1} \ldots A_{i_L} \right) \vert  i_1, \ldots i_N \rangle  \langle j_1, \ldots j_L \vert
\end{equation}
\end{lem}

\begin{proof}
We consider the channel defined as
\begin{equation}
\mathbb{E}(X) = \sum_{i=1}^d V_i X V^{\dagger}_i,
\end{equation} where $V_1 \sqrt{D}$ is a diagonal unitary matrix with different {\it incommensurable} eigenvalues (such that $V_1^k$ still has different eigenvalues for all $k \in \mathbb{N}$), $V_2 \sqrt{D}$ is a random unitary matrix with non-zero entries and $V_i = 0_D, i \in \{3, ... , d\}$. 
This channel is trace preserving and unital. On the one hand, it is trivial to see that the only matrices that commute with $V_1$ are diagonal matrices. On the other hand, to find which of these diagonal matrices commute with $V_2$ it is enough to consider the algebraic system of equations in coordinates for $[V_2, X] = 0$ from where we get that, since $(V_2)_{ij} \neq 0,$ and $(X)_{ij} = 0 \text{ if }i\neq j,$ then $(X)_{ii} - (X)_{jj} = 0$ for all $i \neq j.$ From this, we deduce that the only matrices that commute with all of the Kraus operators for our channel are multiples of the identity matrix. L\"{u}ders' Theorem \cite{luders} guarantees that our channel has the identity as its unique fixed point. Since $\mathbb{E}$ is an irreducible channel \cite{operator-wolf}, all its eigenvalues of modulus $1$ are $k$-roots of unity, where $k \in \{1, \ldots, D^2 \}$. Let $Y$ be such that $\mathbb{E}(Y)=\alpha Y$ for $|\alpha|=1$. It is clear that $\E^k(Y)=Y$ and, again by L\"uders' Theorem, $[V_1^k,Y]=0=[V_2 V_1^{k-1},Y]$. Reasoning as above, $Y$ is a multiple of the identity, which implies that $\alpha=1;$ hence, the channel is primitive \cite{operator-wolf}. 

We can define now
\begin{equation}
\begin{cases}
B_j &= \sqrt{ \Lambda } V_j \\
C_k &= V_k
\end{cases}
\end{equation} where $V_i$ are the Kraus operators for our channel. If we consider the state
\begin{equation}
\vert \psi \rangle = \sum_{\begin{subarray} {l} i_{1}, \ldots, i_L \\  i_{L+1}, \ldots, i_N\end{subarray}} \text{tr} \left( A_{i_1} \ldots A_{i_L} B_{i_{L+1}}  C_{i_{L+2}} \ldots  C_{i_N} \right) \vert i_1, \ldots i_N \rangle
\end{equation} and compute the normalized reduced density matrix for particles $1...L,$ we obtain 
\begin{widetext}
\begin{eqnarray*}
\rho_L = \sum_{\begin{subarray}{l} i_1, \ldots, i_L \\ j_1, \ldots, j_L \end{subarray}} \left( \sum_{\alpha, \beta} \langle \alpha \vert \left[ A_{i_1} \ldots A_{i_L}\mathbb{E}_B \mathbb{E}_C^{N-(L+1)}(\vert \alpha \rangle \langle \beta \vert) A^{\dagger}_{j_L} \ldots A^{\dagger}_{j_1} \right] \vert \beta \rangle \right)
\end{eqnarray*}
\end{widetext}
It is clear that $\mathbb{E}_C^{N-(L+1)} ( \vert \alpha \rangle \langle \beta \vert) = \delta_{\alpha \beta} \mathbbm{1},$ up to an exponentially small correction. This leads us to
\begin{eqnarray*}
\rho_L &=& \sum_{\begin{subarray}{l} i_1, \ldots, i_L \\ j_1, \ldots, j_L \end{subarray}} \left( \sum_{\alpha, \beta} \langle \alpha \vert \left[ A_{i_1} \ldots A_{i_L}\mathbb{E}_B (\1) A^{\dagger}_{j_L} \ldots A^{\dagger}_{j_1} \right] \vert \beta \rangle \right) \\
&=& \sum_{\begin{subarray}{l} i_1, \ldots, i_L \\ j_1, \ldots, j_L \end{subarray}} \left( \sum_{\alpha, \beta} \langle \alpha \vert \left[ A_{i_1} \ldots A_{i_L}\Lambda A^{\dagger}_{j_L} \ldots A^{\dagger}_{j_1} \right] \vert \beta \rangle \right)
\end{eqnarray*} once again, up to a exponentially small correction, just as we wanted to prove.
\end{proof}

\end{document}